\documentclass[11pt,a4paper]{article}
\usepackage{amssymb}
\usepackage{amsmath}
\usepackage{amsfonts}
\usepackage{bbm}
\usepackage{amsthm}
\usepackage{mathrsfs}
\usepackage{hyperref}
\usepackage{color}
\usepackage[margin=2.31cm]{geometry}
\usepackage[all,cmtip]{xy}
\usepackage[utf8]{inputenc}
\usepackage{graphicx}
\usepackage{varwidth}
\usepackage{comment}

\usepackage{upgreek}
\usepackage{rotating}

\usepackage{tikz}
\usetikzlibrary{shapes.geometric}

\definecolor{darkred}{rgb}{0.8,0.1,0.1}
\hypersetup{
     colorlinks=false,         % false: boxed links; true: colored links,false is default
     linkcolor=darkred,
     citecolor=blue,
}

\theoremstyle{plain}
\newtheorem{theo}{Theorem}[section]

\newtheorem{propo}[theo]{Proposition}
\newtheorem{cor}[theo]{Corollary}

\theoremstyle{definition}
\newtheorem{defi}[theo]{Definition}

\newtheorem{remm}[theo]{Remark}
\newenvironment{rem}
  {\pushQED{\qed}\remm}
  {\popQED\endremm}

\numberwithin{equation}{section}

\def\nn{\nonumber}

\def\bbR{\mathbb{R}}

\def\bbZ{\mathbb{Z}}

\def\Imm{\mathrm{Im}}
\def\Ker{\mathrm{Ker}}

\def\id{\mathrm{id}}
\def\supp{\mathrm{supp}}

\def\dd{\mathrm{d}}

\def\sc{\mathrm{sc}}
\def\cc{\mathrm{c}}

\def\dim{\mathrm{dim}}
\def\1{I}

\def\Man{\mathbf{Man}}

\def\Set{\mathbf{Set}}

\def\Ch{\mathbf{Ch}}

\def\BR{\mathbf{B}\bbR}
\def\con{\mathrm{con}}

\def\Grpd{\mathbf{Grpd}}

\def\FFF{\mathfrak{F}}
\def\Sol{\mathfrak{S}}

\def\sk{\vspace{2mm}}

\makeatletter
\let\@fnsymbol\@alph
\makeatother

\title{%
Homological perspective on edge modes\\
in linear Yang-Mills and Chern-Simons theory
}

\author{%
Philippe Mathieu$^{1,a}$, 
Laura Murray$^{1,b}$,
Alexander Schenkel$^{2,c}$\ and\
Nicholas J.\ Teh$^{3,d}$\vspace{4mm}\\
{\small ${}^1$ Department of Mathematics, University of Notre Dame,}\\
{\small  Notre Dame, IN 46556, United States of America.}\vspace{2mm}\\
{\small ${}^2$ School of Mathematical Sciences, University of Nottingham,}\\
{\small University Park, Nottingham NG7 2RD, United Kingdom.}\vspace{2mm}\\
{\small ${}^3$ Department of Philosophy, University of Notre Dame,}\\
{\small  Notre Dame, IN 46556, United States of America.}\vspace{4mm}\\
{\small ${}^a$ \texttt{pmathieu@nd.edu}, ${}^b$ \texttt{lwells@nd.edu}, 
${}^c$~\texttt{alexander.schenkel@nottingham.ac.uk},
${}^d$ \texttt{nteh@nd.edu}}
}

\date{January 2020}

\begin{document}

\maketitle

\vspace{-6mm}

\begin{abstract}
\noindent We provide an elegant homological construction of the extended phase space for linear Yang-Mills theory on an oriented and time-oriented Lorentzian manifold $M$ with a time-like boundary $\partial M$ that was proposed by Donnelly and Freidel [JHEP {\bf 1609}, 102 (2016)]. This explains and formalizes many of the rather ad hoc constructions for edge modes appearing in the theoretical physics literature. Our construction also applies to linear Chern-Simons theory, in which case we obtain the extended phase space introduced by Geiller [Nucl.\ Phys.\ B {\bf 924}, 312 (2017)].
\end{abstract}

\vspace{-2mm}

\vspace{-2mm}

\vspace{-1mm}

%%%%%%%%%%%%%%%%%%%%%%%%%%%%%%%%%%%%%%%%%%%%%%%%
%%%%%%%%%%%%%%%%%%%%%%%%%%%%%%%%%%%%%%%%%%%%%%%%

\section{Introduction and summary}
The topic of edge modes is a time-honored one in the study of gauge theories 
on manifolds with boundary. Historically, such edge modes first arose as the 
(conformal) boundary degrees of freedom of Chern-Simons theory, both in the 
context of Chern-Simons theory as applied to condensed matter physics (see e.g.\
\cite{Balachandran} and \cite{Bieri}), as well as in the context of (3D) Chern-Simons gravity 
\cite{Witten,Carlip}, where the edge modes were shown to be related to the asymptotic 
symmetries of 3D AdS spacetime \cite{Brown,Coussaert,Banados}.
\sk

In these early investigations, the motivation for edge modes and the 
construction of the corresponding boundary action (e.g.\ the Wess-Zumino-Witten 
action and its variants) relied heavily on the fact that the theory's bulk action is not 
gauge-invariant in the presence of a boundary, and the edge modes can heuristically 
be understood as boundary degrees of freedom that `compensate' for this failure of 
gauge-invariance. However, Donnelly and Freidel \cite{DF16} recently showed that one 
can hope to construct edge modes even in cases where the bulk action is gauge-invariant, 
e.g.\ in Yang-Mills theory.  One of their main observations is that, even if the bulk action of a 
gauge theory is gauge-invariant in the presence of boundaries, its
presymplectic form (on field space) may fail to be invariant under arbitrary gauge transformations,
calling for the introduction of boundary-localized degrees of freedom, the edge modes,
to compensate for this lack of invariance. The result of such an analysis
is therefore an {\em extended phase space}, encoding also the additional
edge mode degrees of freedom, that is endowed with a gauge-invariant extension
of the naive presymplectic form by terms depending on the edge modes. This has been
carried out in the original paper \cite{DF16} for the cases of Yang-Mills theory and general relativity.
\sk

Donnelly and Freidel’s work has inspired a revived and growing interest 
in gauge and gravity theories on manifolds with boundaries, see e.g.\ 
\cite{Geiller,Blommaert,Gomes,Gomes2,Gomes3,Gomes4,FP,FLP} for some follow-up papers. 
A particularly noteworthy reaction to their work is \cite{Gomes,Gomes2,Gomes3,Gomes4}, which 
observes that the notion of boundary in \cite{DF16} is ambiguous between 
a `fiducial' boundary, meaning a non-physical boundary that does not 
in any way influence the field content and which disappears upon gluing 
along the boundary, and a `physical' boundary, 
meaning a boundary that influences the field content in some way, 
e.g.\ by carrying a defect theory or a Higgs field.
(This ambiguity is heightened by the fact that \cite{DF16} do not associate 
any action to the edge modes). 
For pure gauge fields, the study in \cite{Gomes2,Gomes4} 
uses a certain Singer-De Witt connection form on field
space, which they interpret as a geometric generalization of ghost fields, 
to show that fiducial boundaries cannot carry charged edge modes. 
Furthermore, in the case where matter fields are present, 
they introduce the notion of a Higgs connection on field space to 
reproduce the edge modes from \cite{DF16}. However, we note that they
also do not introduce a boundary action for these edge modes.
\sk

The goal of this paper is to provide an elegant and rigorous 
construction of extended phase spaces as in \cite{DF16} for two simple cases:
linear Yang-Mills theory on a globally hyperbolic Lorentzian 
manifold $M$ with a time-like boundary $\partial M$,
and linear Chern-Simons theory on a $3$-dimensional product manifold $M=\bbR\times\Sigma$
with boundary $\partial M = \bbR\times\partial\Sigma$.
Our construction employs some basic techniques from homological algebra and 
the theory of groupoids, which are necessary to describe the higher categorical
structures featuring in gauge theory. We refer the reader to 
\cite{Schreiber} for an extensive overview of such techniques
and also to \cite[Section 3]{BSreview} for a rather non-technical introduction.
The main benefit of adopting this more abstract homological perspective 
is that many of the ad hoc constructions for edge modes in the theoretical
physics literature become very natural. 
\sk

The basic ideas behind our proposed construction are easy to explain in general,
without referring to any specific example. Our first input datum is the specification of the 
gauge fields and gauge transformations in the bulk $M$, which assemble into a 
groupoid of bulk gauge fields. As second input, we choose a boundary condition
on $\partial M$ for the bulk gauge fields, which we implement in a
homotopical way by forming a homotopy pullback.
As we explain in detail in Remark \ref{rem:originedgemodes}, see also Remarks
\ref{rem:YM} and \ref{rem:gravity} for further supporting examples, the appearance of edge modes is
a direct consequence of implementing a suitable (topological) boundary condition
in this homotopical fashion. This supports the suggestion in 
\cite{Gomes,Gomes2,Gomes3,Gomes4} that edge modes are associated to physical boundaries.
The last input for our construction
is a gauge-invariant action functional on the total groupoid of fields (including the edge modes)
that is obtained by implementing the boundary condition via a homotopy pullback.
We would like to emphasize that this does not only require the choice of a bulk action, 
but also that of a boundary action, potentially including also terms that depend on the edge modes. 
From this collection of input data, we construct a homotopical refinement of the
solution space, called a derived critical locus, that is associated to our chosen action functional. 
By general results of derived algebraic geometry
\cite{DAG,DAG2,Pridham}, this solution `space' (more precisely, this is a derived stack)
carries a canonical $[-1]$-shifted symplectic structure. From a choice of Cauchy surface $\Sigma\subset M$,
we then determine from the latter data an unshifted symplectic structure on $\Sigma$, hence the 
extended phase space of the theory. Our construction of an unshifted symplectic structure
from the canonical $[-1]$-shifted symplectic structure allows us to carefully distinguish 
between the different types of `boundaries' that feature in our models of interest,
see also \eqref{eqn:integrationpicture} for a helpful visualization.
On the one hand, there is the boundary $\partial M$ on which we impose a boundary 
condition, i.e.\ on which the edge modes are localized, and, on the other hand,
there is the Cauchy surface $\Sigma\subset M$ on which the unshifted symplectic structure
is defined.
\sk

Even though our proposed construction is relatively easy to sketch in an informal way, there are
technical challenges, most notably in the last step where a derived critical
locus and its $[-1]$-shifted (and also unshifted) symplectic structure has to be determined.
Unfortunately, the current technology from derived algebraic geometry \cite{DAG,DAG2,Pridham}
is rather abstract and involved, so that it is very difficult to apply such techniques to examples of
relevance in field theory. In particular, even though derived critical loci always exist 
in this framework, they are very difficult to describe in explicit terms for 
examples such as Yang-Mills theory or general relativity. In order to obtain 
a computationally accessible and feasible
framework, we restrict (drastically!) our attention to the case of {\em linear} gauge theories when discussing
derived critical loci and their symplectic structures. In this case, the necessary techniques
from derived algebraic geometry reduce to relatively basic homological algebra
of chain complexes. We hope that a generalization of this last part of our construction 
to non-linear gauge theories becomes available in the future once the necessary
technology at the intersection of derived algebraic geometry and field theory has been developed.
\sk

The explicit results that we obtain for the simple examples given by
linear Yang-Mills and Chern-Simons theory are however already very interesting. 
For both theories, we make the novel observation that their extended phase 
spaces {\em can} be obtained from simple action functionals (see \eqref{eqn:action}
and \eqref{eqn:CSaction}) via our homological construction, even though this was 
thought to be not possible in \cite{DF16,Geiller}. We believe that our approach via action functionals
is more elegant than the proposal in \cite{DF16} and \cite{Geiller}, which is to introduce by 
hand additional terms to the ordinary symplectic structure
in order to restore gauge-invariance in the presence of a boundary.
As another novel result, our construction leads to an extension of 
the extended phase spaces and their symplectic structures in \cite{DF16,Geiller} 
to ghost fields and antifields, whose explicit form for linear
Yang-Mills theory is given in \eqref{eqn:unshiftedsymplectic} and for
linear Chern-Simons theory in \eqref{eqn:CSsymplectic}. 
\sk

We would like to add a few remarks on the comparison between
our proposed construction and the BV-BFV formalism
for gauge theories on manifolds with boundaries.
This framework originated in \cite{Cattaneo}
and it was extended recently towards
the description of edge mode phenomena in \cite{Mnev}. At a superficial level,
both approaches look similar as they consider, in addition to the gauge fields, 
ghost fields and antifields, and work with shifted symplectic structures.
However, a closer look shows that actual constructions in the BV-BFV formalism 
are performed in a different order than what we propose.
The starting point of \cite{Mnev} is a BV-extended gauge theory 
on a manifold (possibly with boundaries, corners or a stratification),
which however does not yet refer to edge modes and their dynamics.
The latter are obtained from a choice of polarization functional (via an $f$-transformation) 
and an AKSZ-inspired transgression construction, see \cite[Theorem 58]{Mnev} for the 
case of Chern-Simons theory. Interestingly, for appropriate choices of polarization functionals,
this construction produces the gauged Wess-Zumino and gauged Wess-Zumino-Witten
actions for the Chern-Simons edge modes. In contrast to that, the starting point of our construction 
is very basic and it consists of 1.)~a gauge theory (not BV-extended) in the bulk $M$, 
2.)~a boundary condition on $\partial M$ and 3.)~a choice of action functional, 
including possibly also boundary terms on $\partial M$.
The edge modes are then obtained by implementing the boundary condition 
by a homotopy pullback and the BRST/BV field content (with differentials and 
$[-1]$-shifted symplectic structure) is determined from the derived critical locus construction.
We refer to Remarks \ref{rem:CMRlinear} and \ref{rem:MSWcomparison} for a more concrete 
comparison to \cite{Cattaneo,Mnev} at the level of explicit examples.
\sk

The outline of the remainder of this paper is as follows: In Section \ref{sec:model}
we introduce our linear Yang-Mills theory model on a Lorentzian manifold $M$ with a time-like
boundary $\partial M$, together with a boundary condition (leading
to the edge modes, see Remark \ref{rem:originedgemodes}) 
and the novel action functional \eqref{eqn:action}. In Section \ref{sec:shifted} 
we construct explicitly the (linear) derived critical locus for our model
\eqref{eqn:solcomplex} and its canonical $[-1]$-shifted symplectic structure
\eqref{eqn:shiftedsymplectic}. In Section \ref{sec:unshifted} we derive, from the choice
of a Cauchy surface $\Sigma\subset M$, an unshifted symplectic structure \eqref{eqn:unshiftedsymplectic}
and show that the $0$-truncation of our homological construction reproduces
the extended phase space of \cite{DF16}, see Remark \ref{rem:extendedphasespace}.
In Section \ref{sec:CS} we apply our techniques to linear Chern-Simons
theory and show that the $0$-truncation of our construction
reproduces the extended phase space of \cite{Geiller}.
Appendix \ref{app:hPB} summarizes the relevant background for computing
homotopy pullbacks for groupoids and chain complexes that are needed for our work.

\paragraph{Notation and conventions for chain complexes:} The main constructions and results in this paper
are stated in the category $\Ch_\bbR$ of (possibly unbounded) chain complexes 
of vector spaces over the field of real numbers $\bbR$. 
We use homological degree conventions,
i.e.\ the differentials $\dd : V_n\to V_{n-1}$ lower the degree by $1$.
The tensor product $V\otimes W$ of two chain complexes is given by 
$(V\otimes W)_n = \bigoplus_{m\in\bbZ} V_m\otimes W_{n-m}$
together with the differential $\dd(v\otimes w) = (\dd v)\otimes w + 
(-1)^{\vert v\vert} \, v\otimes (\dd w)$ determined by the graded Leibniz rule, 
where $\vert v\vert\in\bbZ$ denotes the degree of $v$. The tensor unit
is $\bbR\in\Ch_\bbR$, regarded as a chain complex concentrated in degree
$0$ with trivial differentials. Given a chain complex $V$ and an integer 
$p\in\bbZ$, the $[p]$-shifted chain complex $V[p]$ is defined by $V[p]_n = V_{n-p}$
and $\dd^{V[p]} \,=\, (-1)^p\,\dd^V$.
\sk

The homology $H_\bullet(V)$ of a chain complex $V$ is the 
graded vector space defined by $H_n(V) := \Ker(\dd : V_n\to V_{n-1}) / \Imm(\dd : V_{n+1}\to V_n)$,
for all $n\in\bbZ$. A chain map $f : V\to W$ is called a quasi-isomorphism
if the induced map $H_\bullet(f) : H_\bullet(V)\to H_\bullet(W)$ in homology is an isomorphism.
Quasi-isomorphic chain complexes are considered as `being the same',
which can be made precise by endowing $\Ch_\bbR$ with a model
category structure, see e.g.\ \cite{Hovey}. We refer to \cite[Section 3]{BSreview}
for a brief non-technical introduction to model categories in the context of
classical and quantum gauge theory.

%%%%%%%%%%%%%%%%%%%%%%%%%%%%%%%%%%%%%%%%%%%%%%%%
%%%%%%%%%%%%%%%%%%%%%%%%%%%%%%%%%%%%%%%%%%%%%%%%

\section{\label{sec:model}Definition of the Yang-Mills model}
Let $M$ be an oriented and time-oriented Lorentzian manifold
with a smooth boundary $\partial M$. 
Following common practice in hyperbolic PDE theory
and Lorentzian (quantum) field theory, we assume that the boundary $\partial M$ is time-like. 
In this case there exists a well-established notion of Cauchy surfaces 
and of global hyperbolicity, see e.g.\ \cite{Solis} and \cite{Ake}. These concepts are not
only important for developing a theory of solutions for
hyperbolic PDEs in the presence of boundaries, see e.g.\ 
\cite{Dappiaggi,Dappiaggi2}, but they will also enter explicitly our 
construction in Section \ref{sec:unshifted}. (We would like to note that the present section and also 
Section \ref{sec:shifted} do not require the assumption of a time-like boundary.)
We denote by $\iota : \partial M\to M$ the boundary inclusion and 
by $m=\dim (M)\geq 2$ the dimension of $M$. 
The orientation, time-orientation and Lorentzian metric 
on $M$ induce on $\partial M$ the structure of an oriented and 
time-oriented Lorentzian manifold (without boundary)
of dimension $\dim(\partial M) =  m-1$.
We interpret $M$ as a physical spacetime whose boundary is another (Lorentzian)
spacetime $\partial M$.
\sk

Let us now introduce the field content of our model of interest.
As bulk fields on $M$ we consider principal $\bbR$-bundles with 
connections, together with their gauge transformations. 
These data are described by the groupoid
\begin{flalign}\label{eqn:BGcon}
\BR_\con(M)\,:=\, \begin{cases}
\mathrm{Obj:} &  A\in \Omega^1(M)\\
\mathrm{Mor:} & A\stackrel{\epsilon}{\longrightarrow} A+\dd\epsilon \quad  \mathrm{with}~~  \epsilon \in \Omega^0(M)
\end{cases}\quad\quad,
\end{flalign}
whose objects are interpreted as gauge fields and morphisms
as gauge transformations between gauge fields.
(Recall that every principal $\bbR$-bundle is isomorphic to the trivial 
principal $\bbR$-bundle. Hence, up to equivalence of groupoids, one may
consider only the trivial principal $\bbR$-bundle, as we have done in \eqref{eqn:BGcon}.)
Take a principal $\bbR$-bundle on the boundary $\partial M$,
which is described by a map of groupoids (i.e.\ a functor)
\begin{flalign}
p\,:\, \{\ast\} ~\longrightarrow~\BR(\partial M)
\end{flalign}
from the point $\{\ast\}$ to the groupoid 
\begin{flalign}
\BR(\partial M)\,:=\, \begin{cases}
\mathrm{Obj:} &  \ast \\
\mathrm{Mor:} & \ast \stackrel{\chi}{\longrightarrow} \ast \quad  \mathrm{with}~~  \chi \in \Omega^0(\partial M)
\end{cases}\quad
\end{flalign}
of principal $\bbR$-bundles on $\partial M$ and their gauge transformations. 
Observe that there is another map of groupoids
\begin{flalign}
\mathrm{res}\,:\, \BR_\con(M)~\longrightarrow~\BR(\partial M)
\end{flalign}
which forgets the bulk connection and restricts the bulk principal $\bbR$-bundle 
to the boundary $\partial M$. Concretely, this functor acts on objects
as $A \mapsto \ast$ and on morphisms as $(\epsilon : A\to A+\dd\epsilon) \mapsto
(\iota^\ast \epsilon : \ast \to \ast)$, where $\iota^\ast \epsilon \in\Omega^0(\partial M)$
denotes the pullback of $\epsilon\in\Omega^0(M)$ along the boundary inclusion
$\iota : \partial M\to M$. We would like to impose a boundary condition
that identifies the restriction of the bulk principal $\bbR$-bundle with the fixed
principal $\bbR$-bundle on $\partial M$. This is formalized by considering 
the {\em homotopy pullback} (or equivalently a $2$-categorical pullback)
\begin{flalign}\label{eqn:boundarycondition}
\xymatrix{
\ar@{-->}[d]\FFF(M) \ar@{-->}[r] &  \ar@{}[dl]_-{h~~~~~} \BR_\con(M)\ar[d]^-{\mathrm{res}}\\
\{\ast\} \ar[r]_-{p}& \BR(\partial M)
}
\end{flalign}
in the model category (or $2$-category) of groupoids, see Appendix \ref{app:hPB}
for some technical details. The resulting groupoid
$\FFF(M)$ plays the role of the groupoid of fields for our model of interest.
\begin{propo}\label{propo:fieldgroupoid}
A model for the homotopy pullback in \eqref{eqn:boundarycondition}
is given by the groupoid
\begin{flalign}\label{eqn:Fgroupoid}
\FFF(M)\,=\, \begin{cases}
\mathrm{Obj:} &  (A,\varphi) \in \Omega^1(M)\times \Omega^0(\partial M)\\
\mathrm{Mor:} & (A,\varphi) \stackrel{\epsilon}{\longrightarrow} \big(A+\dd\epsilon , \varphi + \iota^\ast\epsilon \big) \quad  \mathrm{with}~~  \epsilon \in \Omega^0(M)
\end{cases}\quad\quad.
\end{flalign}
\end{propo}
\begin{proof}
This is a direct computation using the explicit description of homotopy pullbacks
for groupoids from Appendix \ref{app:hPB}, see in particular Proposition \ref{prop:homotopyPBgrpd}.
Concretely, an object in the homotopy pullback \eqref{eqn:boundarycondition} 
is a pair of objects $(\ast,A)\in \{\ast\}\times\BR_\con(M)$
together with a $\BR(\partial M)$-morphism $p(\ast) = \ast \stackrel{\varphi}{\longrightarrow} \ast = \mathrm{res}(A)$.
Hence, an object in $\FFF(M)$ is given by a pair $(A,\varphi) \in \Omega^1(M)\times \Omega^0(\partial M)$.
A morphism $(A,\varphi)\to(A^\prime,\varphi^\prime)$ in the homotopy pullback \eqref{eqn:boundarycondition} 
is a pair of morphisms $(\id_\ast: \ast\to \ast, \epsilon: A\to A+\dd \epsilon = A^\prime)\in  \{\ast\}\times\BR_\con(M)$
that is compatible with $\varphi$ and $\varphi^\prime$, i.e.\ the diagram
\begin{flalign}
\xymatrix@C=2.5em{
p(\ast) = \ast \ar[d]_-{\varphi} \ar[r]^-{\id_\ast}& \ast = p(\ast) \ar[d]^-{\varphi^\prime}\\
\mathrm{res}(A)=\ast \ar[r]_-{\iota^\ast \epsilon} & \ast = \mathrm{res}(A^\prime)
}
\end{flalign}
in $\BR(\partial M)$ commutes. Hence, a morphism in $\FFF(M)$ is given by
$(A,\varphi) \stackrel{\epsilon}{\longrightarrow} (A+\dd\epsilon , \varphi + \iota^\ast\epsilon )$, 
where $\epsilon \in \Omega^0(M)$.
\end{proof}

\begin{rem}\label{rem:originedgemodes}
Note that an object of the groupoid $\FFF(M)$ in \eqref{eqn:Fgroupoid} 
is a pair $(A,\varphi)\in\Omega^1(M)\times \Omega^0(\partial M)$ consisting of
a gauge field $A$ in the bulk $M$ and a gauge transformation $\varphi$ on the boundary $\partial M$.
Hence, the groupoid of fields $\FFF(M)$ contains both bulk and boundary fields.
It is one of the main goals of the present paper to explain that these $\varphi$ 
are precisely the edge modes introduced in \cite{DF16}. As a first piece of evidence
for this claim, we note that the morphisms of the groupoid $\FFF(M)$ in \eqref{eqn:Fgroupoid} 
are precisely the gauge transformations on bulk and boundary fields in \cite{DF16}.
\sk

From our groupoid perspective, the origin of edge modes can be 
explained very naturally. The groupoid of fields $\FFF(M)$ is obtained by identifying
the restriction of the bulk principal $\bbR$-bundle with the fixed
principal $\bbR$-bundle on $\partial M$, i.e.\ we implement 
a boundary condition via the homotopy pullback diagram \eqref{eqn:boundarycondition}. 
We may call this boundary condition {\em topological} because,
in contrast to the usual Dirichlet or Neumann-type boundary conditions,
it only involves the underlying principal bundle and not the connection part of
the bulk gauge field. Boundary conditions in a gauge theory are quite subtle because gauge fields are {\em not} 
compared by equality but rather by gauge transformations, 
i.e.\ morphisms in the relevant groupoids. Hence, a boundary condition
in a gauge theory is not a {\em property} of the gauge fields but an 
additional {\em structure} given by gauge transformations acting as witnesses
of the boundary condition. The edge modes $\varphi$ in \eqref{eqn:Fgroupoid}
are precisely the witnesses for the statement that the restriction of the bulk principal $\bbR$-bundle
is `the same as' the fixed boundary principal $\bbR$-bundle.
\end{rem}

In the next step we introduce a gauge-invariant 
action functional in order to specify the dynamics of our model of interest. 
This is described by a map of groupoids $S : \FFF(M)\to \bbR$
from our groupoid of fields \eqref{eqn:Fgroupoid} to the real numbers $\bbR$, 
regarded as a groupoid with only identity morphisms. We define
\begin{flalign}\label{eqn:action}
S(A,\varphi) \,:=\, \int_M \frac{1}{2} \, \dd A\wedge\ast\dd A + 
\int_{\partial M} \frac{1}{2}\,  \dd_A \varphi \wedge \ast_{\partial}^{} \dd_A \varphi\quad,
\end{flalign}
where $\ast_{(\partial)}^{}$ denotes the Hodge operator on $(\partial)M$
and the affine covariant differential is given by
\begin{flalign}
\dd_A \varphi\,:=\, \dd \varphi - \iota^\ast A\quad.
\end{flalign}
Clearly, the action (\ref{eqn:action}) is gauge-invariant because $\dd A$ and $\dd_A\varphi$ 
are invariant under the gauge transformations in \eqref{eqn:Fgroupoid}.
(In the physics literature, the quantity $\dd_A \varphi$ is also referred to as a `dressing', see e.g.\ \cite{Dressing}.)
\sk

Upon varying the action with respect to compactly supported variations 
$(\alpha,\psi)\in \Omega^1_\cc(M)\times \Omega^0_\cc(\partial M)$,
a straightforward calculation using Stokes' theorem yields the expression 
\begin{flalign}\label{eqn:variationaction}
\delta_{(\alpha,\psi)}S(A,\varphi) \,=\, \int_M \alpha\wedge \dd \ast\dd A + 
\int_{\partial M}  \Big( \iota^\ast\alpha \wedge \big(\iota^\ast(\ast\dd A) - \ast_{\partial}^{} \dd_{A}\varphi\big)
-  \psi\wedge \dd\ast_{\partial}^{} \dd_A \varphi\Big)\quad.
\end{flalign}
The corresponding Euler-Lagrange equations are
\begin{subequations}\label{eqn:ELequations}
\begin{flalign}
\dd\ast\dd A &\,=\,0 \qquad\text{(linear Yang-Mills equation on $M$)}\quad,\\
\dd\ast_{\partial}^{}\dd_A \varphi &\,=\, 0 \qquad \text{(inhomogeneous Klein-Gordon equation 
on $\partial M$)}\quad,\\
\iota^\ast(\ast\dd A) - \ast_\partial^{} \dd_{A}\varphi&\,=\,0\qquad
 \text{(matching constraint on $\partial M$)}\quad.
\end{flalign}
\end{subequations}
\begin{rem}\label{rem:action}
We would like to emphasize very clearly that both the bulk and boundary terms in the action 
\eqref{eqn:action} are inputs for our construction that one has to choose. 
Besides its evident simplicity, our choice of action is motivated from the fact
that its Euler-Lagrange equations \eqref{eqn:ELequations} include
the matching constraint, which has been implemented by hand 
in the work of Donnelly and Freidel \cite{DF16}. Of course, it would be possible
to choose a different action, for example by introducing a multiplicative 
factor $\lambda\in\bbR$ in front of the boundary term in \eqref{eqn:action},
which would lead to different Euler-Lagrange equations, including a different
matching constraint between bulk and boundary fields. The constructions
that we develop in this paper apply to general gauge-invariant quadratic 
actions functionals. However, our focus will be on the action \eqref{eqn:action} 
because our main aim is to reconstruct and interpret the model of \cite{DF16} 
from a homological point of view.
\end{rem}

\begin{rem}\label{rem:YM}
Up to this point, our construction 
admits a straightforward generalization to non-Abelian Yang-Mills theory.
To simplify the presentation in this remark, let us work locally by
assuming that $M\cong \bbR^{m-1}\times [0,\infty)$
is diffeomorphic to a half-space. Let $G$ be a compact matrix Lie group and denote its Lie algebra by $\mathfrak{g}$.
As a consequence of our assumptions, there exist no non-trivial principal $G$-bundles on both $M$ and $\partial M$, 
hence the groupoid of principal $G$-bundles with connection on $M$ reads as
\begin{flalign}
\mathbf{B}G_\con(M)\,:=\, \begin{cases}
\mathrm{Obj:} &  A\in \Omega^1(M,\mathfrak{g})\\
\mathrm{Mor:} & A\stackrel{g}{\longrightarrow} g^{-1}Ag+g^{-1}\dd g \quad  \mathrm{with}~~  g \in C^\infty(M,G)
\end{cases}\quad
\end{flalign}
and the groupoid of principal $G$-bundles on $\partial M$ reads as
\begin{flalign}
\mathbf{B}G(\partial M)\,:=\, \begin{cases}
\mathrm{Obj:} &  \ast \\
\mathrm{Mor:} & \ast \stackrel{h}{\longrightarrow} \ast \quad  \mathrm{with}~~  h \in C^\infty(\partial M,G)
\end{cases}\quad.
\end{flalign}
The two maps in the homotopy pullback diagram \eqref{eqn:boundarycondition}
exist also in the non-Abelian setting. An explicit computation as in Proposition \ref{propo:fieldgroupoid}
yields the groupoid of fields
\begin{flalign}
\FFF_G(M)\,=\, \begin{cases}
\mathrm{Obj:} &  (A,u) \in \Omega^1(M,\mathfrak{g})\times C^\infty(\partial M,G)\\
\mathrm{Mor:} & (A,u) \stackrel{g}{\longrightarrow} \big(g^{-1}Ag+g^{-1}\dd g , u\,\iota^\ast g \big) \quad  \mathrm{with}~~  g \in C^\infty(M,G)
\end{cases}\quad\quad.
\end{flalign}
Recalling the curvature $F(A) = \dd A + A\wedge A$ and introducing the non-Abelian `dressing'
$\dd_{A} u := (\dd u)\,u^{-1} - u\,(\iota^\ast A)\, u^{-1}$, one easily checks that
\begin{flalign}\label{eqn:actionnonabelian}
S_G(A,u) \,:=\, \int_M\frac{1}{2}\, \mathrm{Tr}\big(F(A)\wedge \ast F(A)\big) + 
\int_{\partial M} \frac{1}{2}\,\mathrm{Tr}\big(\dd_A u \wedge \ast_\partial \dd_A u\big)
\end{flalign}
defines a gauge-invariant action. The Euler-Lagrange equations
of this action include the matching constraint $u^{-1}\,\iota^\ast (\ast F(A)) \,u - \ast_{\partial} \dd_A u =0$
on $\partial M$, which agrees with the proposal in \cite{DF16} by introducing
the notation $E:= \ast_{\partial} \dd_A u \in \Omega^{m-2}(\partial M,\mathfrak{g})$.
We will not develop this non-Abelian generalization
of our model any further, because linearity will be crucial to simplify our constructions 
in the remainder of this paper.
\end{rem}

\begin{rem}\label{rem:gravity}
We would like to comment very briefly on the gravity example considered in \cite{DF16}.
The origin of the gravitational edge modes may be understood in terms of
a boundary condition too. Let us fix as in \cite{DF16}
an $m$-dimensional manifold  $M$ with smooth boundary $\partial M$.
The bulk fields are given by the groupoid $\mathbf{Lor}_m(M)$ whose
\begin{itemize}
\item[(i)] objects are all (Lorentzian) metrics $g$ on $M$, and
\item[(ii)] morphisms $f : (M,g)\to (M,g^\prime)$ are all diffeomorphisms 
$f: M\to M$ preserving the boundary, i.e.\ the restriction
$f_\partial := f\vert_{\partial M} : \partial M\to \partial M$ is a diffeomorphism, 
and the metrics, i.e.\ $f^\ast (g^\prime) = g$ holds true.
\end{itemize}
Let us denote by $\Man_{m-1}$ the groupoid of $m-1$-dimensional manifolds,
with morphisms given by diffeomorphisms. 
There exists an evident functor $\mathrm{res} : \mathbf{Lor}_m(M)\to \Man_{m-1}$ 
acting on objects as $(M,g)\mapsto \partial M$ and on morphisms as 
$(f : (M,g)\to (M,g^\prime) ) \mapsto (f_\partial : \partial M\to \partial M)$.
Choosing any object $B \in \Man_{m-1}$ that is diffeomorphic to $\partial M\in \Man_{m-1}$, which
we may regard as a functor $B : \{\ast\}\to \Man_{m-1}$, we can form the homotopy pullback
\begin{flalign}
\xymatrix{
\ar@{-->}[d]\FFF_{\mathrm{gravity}}(M) \ar@{-->}[r] &  \ar@{}[dl]_-{h~~~~~} \mathbf{Lor}_m(M)\ar[d]^-{\mathrm{res}}\\
\{\ast\} \ar[r]_-{B}& \Man_{m-1}
}
\end{flalign}
which implements the metric-independent boundary condition that 
the boundary $\partial M$ of the bulk manifold $M$ is `the same as' the fixed $m-1$-dimensional
manifold $B$. Computing this homotopy pullback via Proposition \ref{prop:homotopyPBgrpd}, we obtain
that
\begin{itemize}
\item[(i)] objects in $\FFF_{\mathrm{gravity}}(M)$ are pairs
$(g, X )$, where  $g$ is a metric on $M$ and $X : B\to \partial M$ is a diffeomorphism
between $B$ and the boundary $\partial M$, and
\item[(ii)] morphisms $f: (g, X )\to (g^\prime,X^\prime)$ in $\FFF_{\mathrm{gravity}}(M)$ are all 
diffeomorphisms $f : M\to M$ preserving the metrics, i.e.\ $f^\ast(g^\prime) =g$, and satisfying
$f_\partial \circ X = X^\prime$.
\end{itemize}
The diffeomorphisms $X:B \to \partial M$ are the gravitational edge modes from \cite{DF16}.
We currently do not know a gravitational analog of the boundary 
actions in \eqref{eqn:actionnonabelian} and \eqref{eqn:action},
but we believe that it should be possible to design such an action 
by carefully studying the gravitation matching constraints in \cite{DF16}. 
As for the case of non-Abelian Yang-Mills theory from the previous remark,
we will not develop this gravity model any further, because our remaining
constructions require a linear field theory.
\end{rem}

Because our model of interest is a linear gauge theory, we can reformulate it 
in the language of chain complexes of vector spaces. The key ingredient
for this construction is given by the Dold-Kan correspondence between simplicial
vector spaces and (non-negatively graded) chain complexes of vector spaces, see e.g.\ \cite{BSShocolim}
for an application in the context of gauge theory. 
Explicitly, the Dold-Kan correspondence assigns to (the nerve of) our groupoid of fields 
\eqref{eqn:Fgroupoid} the chain complex (denoted with abuse of notation
by the same symbol)
\begin{subequations}\label{eqn:fieldcomplex}
\begin{flalign}
\FFF(M) \,=\, 
\Big(
\xymatrix@C=2.5em{
\stackrel{(0)}{\FFF_0(M)} & \ar[l]_-{Q} \stackrel{(1)}{\FFF_1(M)}
}
\Big)
\,=\,
\Big(
\xymatrix@C=2.5em{
\stackrel{(0)}{\Omega^1(M)\times\Omega^0(\partial M)} & \ar[l]_-{Q} \stackrel{(1)}{\Omega^0(M)}
}
\Big)
\end{flalign}
concentrated in homological degrees $0$ and $1$, with differential given by
\begin{flalign}
Q(C)\,=\, \big(\dd C,\iota^\ast C\big)\quad,
\end{flalign}
\end{subequations}
for all $C\in \Omega^0(M)$. 
From now on, we shall denote gauge transformations by $C\in \Omega^0(M)$. This choice of notation
is explained in Remark \ref{rem:BRSTBV} below, where $C$ will be interpreted as a ghost field.
Observe that elements $(A,\varphi)\in\Omega^1(M)\times \Omega^0(\partial M)$
in degree $0$ are the fields of the theory, elements $C\in\Omega^0(M)$ in degree $1$ are the
gauge transformations and the differential $Q$ encodes the action 
$(A,\varphi)\to (A,\varphi) + Q(C) = (A+\dd C,\varphi + \iota^\ast C)$ of gauge transformations.
The variation of the action \eqref{eqn:variationaction} determines a linear differential operator
\begin{subequations}\label{eqn:Poperator}
\begin{flalign}
P \,:\, \Omega^1(M)\times \Omega^0(\partial M)~\longrightarrow~\Omega^{m-1}(M)\times \Omega^{m-2}(\partial M) \times \Omega^{m-1}(\partial M)
\end{flalign} 
given by
\begin{flalign}
P(A,\varphi) \,=\, \Big((-1)^{m-1}\,\dd\ast\dd A , (-1)^{m-2}\,\big(\iota^\ast(\ast\dd A) - \ast_\partial^{} \dd_A \varphi \big), 
- \dd\ast_{\partial}^{}\dd_A \varphi\Big)\quad,
\end{flalign}
\end{subequations}
for all $(A,\varphi)\in \Omega^1(M)\times \Omega^0(\partial M)$. The signs in \eqref{eqn:Poperator}
are due to the following choice of conventions:
The codomain of $P$ is given by the smooth Lefschetz dual
\begin{subequations}\label{eqn:FFF0dual}
\begin{flalign}
\FFF_{0,\cc}(M)^\ast \,:=\,  \Omega^{m-1}(M)\times \Omega^{m-2}(\partial M) \times \Omega^{m-1}(\partial M)
\end{flalign}
of the degree $0$ component $\FFF_{0,\cc}(M)= \Omega^1_\cc(M)\times \Omega^0_\cc(\partial M)$
of the compactly supported analog of the field complex \eqref{eqn:fieldcomplex}.
The evaluation pairing $\langle\,\cdot\, ,\,\cdot\,\rangle  : \FFF_{0,\cc}(M)^\ast\times \FFF_{0,\cc}(M)\to \bbR$ reads as
\begin{flalign}
\langle (A^\dagger, a^\dagger,\varphi^\dagger), (A,\varphi)\rangle = 
\int_M A^\dagger \wedge A  + \int_{\partial M}\big( a^\dagger \wedge \iota^\ast A + \varphi^\dagger\wedge \varphi \big)\quad,
\end{flalign}
\end{subequations}
for all $(A^\dagger, a^\dagger,\varphi^\dagger) \in \Omega^{m-1}(M)\times \Omega^{m-2}(\partial M) 
\times \Omega^{m-1}(\partial M)$ and $(A,\varphi)\in \Omega^1_\cc(M)\times\Omega^0_\cc(\partial M)$.
The linear differential operator $P$ is defined by \eqref{eqn:variationaction}
and the equation $\delta_{(\alpha,\psi)}S(A,\varphi) = \langle P(A,\varphi),(\alpha,\psi)\rangle$, 
for all $(A,\varphi)\in \Omega^1(M)\times \Omega^0(\partial M)$ and 
$(\alpha,\psi)\in \Omega_\cc^1(M)\times \Omega_\cc^0(\partial M)$.
Hence, the signs in \eqref{eqn:Poperator} 
are a consequence of graded commutativity of the $\wedge$-product.

%%%%%%%%%%%%%%%%%%%%%%%%%%%%%%%%%%%%%%%%%%%%%%%%
%%%%%%%%%%%%%%%%%%%%%%%%%%%%%%%%%%%%%%%%%%%%%%%%

\section{\label{sec:shifted}Derived critical locus and shifted symplectic structure}
Instead of enforcing the Euler-Lagrange equations \eqref{eqn:ELequations} 
in a strict sense, we consider their homological enhancement given by the 
(linear) derived critical locus construction. 
Our motivation and reasons for this are twofold: 1.)~Enforcing the Euler-Lagrange equations
strictly as in \eqref{eqn:ELequations} is in general incompatible with
quasi-isomorphisms in the category $\Ch_\bbR$ of (possibly unbounded) chain complexes, i.e.\
if one takes two different quasi-isomorphic field complexes, the naive solution complexes
assigned to them are in general no longer quasi-isomorphic.
This is problematic because it violates the main principle of homological
algebra that all sensible constructions must respect quasi-isomorphisms.
2.)~Every derived critical locus carries a canonical 
$[-1]$-shifted symplectic structure (see e.g.\  \cite{DAG,DAG2,Pridham}
for the corresponding results in derived algebraic geometry) which has various physical applications.
For instance, in the context of (quantum) field theory, this shifted symplectic structure is the starting
point for constructing a factorization algebra \cite{CostelloGwilliam}
or an algebraic quantum field theory  \cite{BSShYM}.
Below, we give a novel application of this $[-1]$-shifted symplectic structure:
It will be used to construct the extended phase space introduced in \cite{DF16}.
We note that in physics terminology, derived critical loci are called the BRST/BV formalism
and the shifted symplectic structure is called the antibracket.
\sk

Our construction of the (linear) derived critical locus and its shifted symplectic
structure is a relatively straightforward generalization of the case of linear 
Yang-Mills theory on spacetimes without boundaries presented in \cite{BSShYM,BSreview}. 
To make the present paper self-contained, we shall briefly explain this construction. 
By analogy with \eqref{eqn:FFF0dual}, we define the smooth Lefschetz dual
\begin{subequations}\label{eqn:FFF1dual}
\begin{flalign}
\FFF_{1,\cc}(M)^\ast \,:=\,  \Omega^m(M)\times \Omega^{m-1}(\partial M)
\end{flalign}
of the degree $1$ component $\FFF_{1,\cc}(M)=\Omega^0_\cc(M)$ of the compactly supported 
analog of the field complex \eqref{eqn:fieldcomplex}. The evaluation
pairing $\langle\,\cdot\, ,\,\cdot\,\rangle  : \FFF_{1,\cc}(M)^\ast\times \FFF_{1,\cc}(M)\to \bbR$ reads as
\begin{flalign}
\langle (C^\dagger, c^\dagger), C\rangle = 
\int_M C^\dagger \wedge C + \int_{\partial M} c^\dagger \wedge \iota^\ast C \quad,
\end{flalign}
\end{subequations}
for all $(C^\dagger, c^\dagger) \in \Omega^m(M)\times \Omega^{m-1}(\partial M)$
and $C\in \Omega^0_\cc(M)$. We denote by 
\begin{subequations}\label{eqn:Qast}
\begin{flalign}
Q^\ast\,:\, \Omega^{m-1}(M)\times \Omega^{m-2}(\partial M) \times \Omega^{m-1}(\partial M)~\longrightarrow~
\Omega^m(M)\times \Omega^{m-1}(\partial M)
\end{flalign} 
the formal adjoint of the linear differential operator $Q$ in  \eqref{eqn:fieldcomplex}, which is defined implicitly
by $\langle Q^\ast(A^\dagger,a^\dagger,\varphi^\dagger), C \rangle 
= \langle (A^\dagger,a^\dagger,\varphi^\dagger), Q(C) \rangle$, for all $(A^\dagger,a^\dagger,\varphi^\dagger)\in 
\Omega^{m-1}(M)\times \Omega^{m-2}(\partial M) \times \Omega^{m-1}(\partial M)$
and $C\in\Omega^0_\cc(M)$. A straightforward calculation using
Stokes' theorem then provides the explicit expression
\begin{flalign}
Q^\ast(A^\dagger,a^\dagger,\varphi^\dagger)\,=\,
\Big((-1)^m\,\dd A^\dagger , (-1)^{m-1}\,\big(\dd a^\dagger + \iota^\ast A^\dagger\big) + \varphi^\dagger\Big)\quad,
\end{flalign}
for all $(A^\dagger,a^\dagger,\varphi^\dagger)\in 
\Omega^{m-1}(M)\times \Omega^{m-2}(\partial M) \times \Omega^{m-1}(\partial M)$.
The smooth Lefschetz dual of the compactly supported analog of the field complex
\eqref{eqn:fieldcomplex} is thus given by
\begin{flalign}
\nn \FFF_\cc(M)^\ast \,&=\, \Big(
\xymatrix@C=2.5em{
\stackrel{(-1)}{\FFF_{1,\cc}(M)^\ast } & \ar[l]_-{-Q^\ast} \stackrel{(0)}{\FFF_{0,\cc}(M)^\ast }
}
\Big) \\
\,&=\,
\Big(
\xymatrix@C=2.5em{
\stackrel{(-1)}{\Omega^m(M)\times \Omega^{m-1}(\partial M)} & \ar[l]_-{-Q^\ast} \stackrel{(0)}{\Omega^{m-1}(M)\times \Omega^{m-2}(\partial M) \times \Omega^{m-1}(\partial M)}
}
\Big)\quad.
\end{flalign}
\end{subequations}
This chain complex is used to define the total space
$T^\ast \FFF(M) := \FFF(M)\times \FFF_\cc(M)^\ast \in\Ch_\bbR$
of the cotangent bundle over  the field complex \eqref{eqn:fieldcomplex} 
as a Cartesian product of chain complexes. The variation of the action 
\eqref{eqn:variationaction}, or equivalently, the associated differential operator $P$ in
\eqref{eqn:Poperator}, defines a section
\begin{flalign}
\parbox{0.5cm}{\xymatrix{
\FFF(M)\ar[d]_-{\delta S}\\
T^\ast\FFF(M)
}
}~~=~~~~ \left(\parbox{2cm}{\xymatrix@C=2.5em{
\ar[d]_-{0}0  &\ar[l]_-{0} \ar[d]_-{(\id, P )}\FFF_0(M) & \ar[l]_-{Q} \FFF_1(M)\ar[d]_-{\id} \\
\FFF_{1,\cc}(M)^\ast &\ar[l]^-{-Q^\ast \pi_2} \FFF_0(M)\times \FFF_{0,\cc}(M)^\ast & \ar[l]^-{\iota_1 Q} \FFF_1(M)
}}\right)\quad
\end{flalign}
of the cotangent bundle. The zero-section of the cotangent bundle is given by
\begin{flalign}
\parbox{0.5cm}{\xymatrix{
\FFF(M)\ar[d]_-{0}\\
T^\ast\FFF(M)
}
}~~=~~~~ \left(\parbox{2cm}{\xymatrix@C=2.5em{
\ar[d]_-{0}0  &\ar[l]_-{0} \ar[d]_-{(\id, 0 )}\FFF_0(M) & \ar[l]_-{Q} \FFF_1(M)\ar[d]_-{\id} \\
\FFF_{1,\cc}(M)^\ast &\ar[l]^-{-Q^\ast \pi_2} \FFF_0(M)\times \FFF_{0,\cc}(M)^\ast & \ar[l]^-{\iota_1 Q} \FFF_1(M)
}}\right)\quad.
\end{flalign}
In order to enforce the dynamics encoded by the action functional \eqref{eqn:action},
we intersect $\delta S$ with the zero-section $0$ (in the derived sense) by forming the homotopy pullback
\begin{flalign}\label{eqn:derivedcriticallocus}
\xymatrix{
\ar@{-->}[d]\Sol(M) \ar@{-->}[r] &  \ar@{}[dl]_-{h~~~~~} \FFF(M)\ar[d]^-{\delta S}\\
\FFF(M) \ar[r]_-{0}& T^\ast \FFF(M)
}
\end{flalign}
in the model category $\Ch_\bbR$.
\begin{propo}\label{propo:solcomplex}
A model for the homotopy pullback in \eqref{eqn:derivedcriticallocus} is given by
the chain complex
\begin{flalign}\label{eqn:solcomplex}
\Sol(M) \,=\, \Big(
\xymatrix@C=2.5em{
\stackrel{(-2)}{\FFF_{1,\cc}(M)^\ast}
& \ar[l]_-{Q^\ast} \stackrel{(-1)}{\FFF_{0,\cc}(M)^\ast}
& \ar[l]_-{P } \stackrel{(0)}{\FFF_{0}(M)}
& \ar[l]_-{Q} \stackrel{(1)}{\FFF_{1}(M)}
}
\Big)\quad,
\end{flalign}
with differentials defined in \eqref{eqn:fieldcomplex}, \eqref{eqn:Poperator} and \eqref{eqn:Qast}.
\end{propo}
\begin{proof}
This is a direct consequence of the explicit description of homotopy pullbacks for chain complexes
from Appendix \ref{app:hPB}, see in particular Proposition \ref{prop:homotopyPBchain}.
In the present scenario we have that $V = \FFF(M)$ is the field complex \eqref{eqn:fieldcomplex}, 
$W = \FFF_\cc(M)^\ast$ is the smooth Lefschetz dual \eqref{eqn:Qast} and 
$f_0:V_0\to W_0$  is the differential operator $P$ in \eqref{eqn:Poperator}.
Inserting this into \eqref{eqn:LcomplexTMP} yields \eqref{eqn:solcomplex}.
\end{proof}

\begin{rem}\label{rem:BRSTBV}
The chain complex \eqref{eqn:solcomplex} admits an interpretation in terms of the BRST/BV formalism.
Elements $C\in\Sol_1(M)=\Omega^0(M)$ in degree $1$ are the ghost fields and elements 
$(A,\varphi)\in\Sol_0(M)=\Omega^1(M)\times\Omega^0(\partial M)$ in degree $0$ are
the fields of the theory. Furthermore, elements $(A^\dagger,a^\dagger,\varphi^\dagger)\in \Sol_{-1}(M)=
\Omega^{m-1}(M)\times\Omega^{m-2}(\partial M)\times \Omega^{m-1}(\partial M)$
in degree $-1$ are the antifields and elements $(C^\dagger,c^\dagger)\in 
\Omega^m(M)\times\Omega^{m-1}(\partial M)$
in degree $-2$ are the antifields for ghosts. The differential operator $Q$ encodes the gauge symmetries
and $P$ encodes the equation of motion of our model. In particular,
the $0$-th homology $H_0(\Sol(M))$ of  \eqref{eqn:solcomplex} 
is the ordinary vector space of gauge equivalence classes of solutions of the Euler-Lagrange equations
\eqref{eqn:ELequations}. Note that, in contrast to the usual BRST/BV formalism on manifolds
without a boundary, our model of interest \eqref{eqn:solcomplex} also contains
boundary fields $\varphi$ and boundary antifields $a^\dagger$, $\varphi^\dagger$ and $c^\dagger$.
It is important to emphasize that this field content is not arbitrary, but it is dictated 
(up to quasi-isomorphism) by our homological approach, i.e.\ by the 
homotopy pullbacks in \eqref{eqn:boundarycondition} and \eqref{eqn:derivedcriticallocus}.
\end{rem}

To conclude this section, we explicitly write out the canonical $[-1]$-shifted symplectic structure
that exists on the (linear) derived critical locus \eqref{eqn:derivedcriticallocus}. 
Denoting by $\Sol_\cc(M)$ the compactly supported analog of the solution complex
$\Sol(M)$ in  \eqref{eqn:solcomplex}, the $[-1]$-shifted symplectic
structure is the chain map $\omega_{-1} : \Sol_\cc(M)\otimes \Sol_\cc(M)\to \bbR[-1]$
defined in terms of the integration pairings \eqref{eqn:FFF0dual} and \eqref{eqn:FFF1dual} by
\begin{subequations}\label{eqn:shiftedsymplectic}
\begin{flalign}
\omega_{-1}\big((C^\dagger,c^\dagger), C \big)\,&=\, 
\int_M C^\dagger\wedge C + \int_{\partial M}c^\dagger \wedge \iota^\ast C\quad,\\
\omega_{-1}\big(C,(C^\dagger,c^\dagger) \big)\,&=\, - \omega_{-1}\big((C^\dagger,c^\dagger), C \big)\quad,\\
\omega_{-1}\big((A^\dagger,a^\dagger,\varphi^\dagger),(A,\varphi)\big)\,&=\,
 \int_M A^\dagger\wedge A + \int_{\partial M} \big(a^\dagger\wedge \iota^\ast A + \varphi^\dagger\wedge\varphi\big)\quad,\\ 
\omega_{-1}\big((A,\varphi),(A^\dagger,a^\dagger,\varphi^\dagger)\big)\,&=\, -
 \omega_{-1}\big((A^\dagger,a^\dagger,\varphi^\dagger),(A,\varphi)\big)\quad,
\end{flalign}
\end{subequations}
for all $(C^\dagger,c^\dagger)\in \Omega_\cc^m(M)\times \Omega_{\cc}^{m-1}(\partial M)$,
$C \in \Omega^0_\cc(M)$, $(A^\dagger,a^\dagger,\varphi^\dagger)\in \Omega_\cc^{m-1}(M)
\times \Omega_{\cc}^{m-2}(\partial M) 
\times \Omega_{\cc}^{m-1}(\partial M)$ and $(A,\varphi)\in \Omega^1_\cc(M)\times\Omega^0_{\cc}(\partial M)$.

%%%%%%%%%%%%%%%%%%%%%%%%%%%%%%%%%%%%%%%%%%%%%%%%
%%%%%%%%%%%%%%%%%%%%%%%%%%%%%%%%%%%%%%%%%%%%%%%%

\section{\label{sec:unshifted}Construction of the unshifted symplectic structure}
From now on, we assume that $M$ is globally hyperbolic in the sense of Lorentzian manifolds
with a time-like boundary, see e.g.\ \cite{Solis}, \cite{Ake}  and also \cite{BDSboundary} for a review.
Let us choose any Cauchy surface $\Sigma\subset M$ and note that $\Sigma$ is
a manifold with boundary $\partial \Sigma \subset \partial M$.
The aim of this section is to construct from the datum of a Cauchy surface
$\Sigma\subset M$ and the $[-1]$-shifted symplectic structure $\omega_{-1}$
in \eqref{eqn:shiftedsymplectic} an unshifted symplectic structure $\omega^\Sigma_0$.
We will then show that the extended phase space proposed by Donnelly and Freidel in \cite{DF16}
is given by the $0$-truncation of this homological construction. 
\sk

Before we can state our definition of the unshifted symplectic structure $\omega^\Sigma_0$,
we will need to introduce some simple concepts from Lorentzian geometry.
Let us denote by 
\begin{subequations}
\begin{flalign}
\Sigma_+ \,:=\, J^+_M(\Sigma)\,\subseteq\, M
\end{flalign}
the causal future of the Cauchy surface $\Sigma\subset M$, which is the set of all
points $p\in M$ that can be reached from $\Sigma\subset M$ 
via future-pointing causal curves, including all points $p\in\Sigma$ in the Cauchy surface.
Note that, by definition, $\Sigma\subset \Sigma_+$ is a subset. We denote by
\begin{flalign}
(\partial \Sigma)_+ \,:=\, \Sigma_+\cap \partial M\,\subseteq \partial M
\end{flalign}
\end{subequations}
the intersection of $\Sigma_+$ with the boundary of $M$.
The following picture visualizes our geometric setup
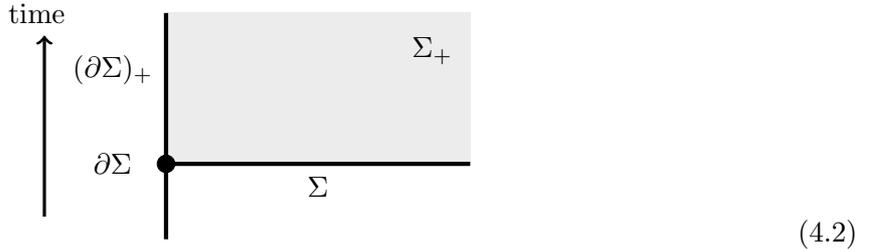
\begin{flalign}\label{eqn:integrationpicture}
\begin{tikzpicture}[scale=1]
\draw[very thick, ->] (-1.6,-0.7) -- (-1.6,1.7);
\draw (-1.7,2) node {\text{time}};
\fill [gray!15] (0,0) rectangle (4,2);
\draw [ultra thick] (0,-1) -- (0,2);
\draw (-0.7,1.25) node {\text{$(\partial \Sigma)_+$}};
\draw [ultra thick] (0,0) -- (4,0);
\draw (2,-0.3) node {\text{$\Sigma$}};
\draw (3.5,1.5) node {\text{$\Sigma_+$}};
\draw [fill=black,draw=black] (0,0) circle (0.3em);
\draw (-0.7,0) node {\text{$\partial \Sigma$}};
\end{tikzpicture}
\end{flalign}
We observe that $\Sigma_+$ has two different kinds of boundary components, 
given by the time-like boundary $(\partial \Sigma)_+$ and the (space-like) 
Cauchy surface $\Sigma$, as well as a codimension $2$ corner $\partial \Sigma$. 
\sk

We now define a map $\omega_{-1}^{\Sigma} : \Sol_\cc(M)\otimes\Sol_\cc(M)\to\bbR[-1]$
of graded vector spaces by recalling the definition of the $[-1]$-shifted symplectic 
structure in \eqref{eqn:shiftedsymplectic}
and restricting the integrations therein from $M$ to $\Sigma_+$ and from $\partial M$ to $(\partial \Sigma)_+$.
Explicitly, this gives 
\begin{subequations}\label{eqn:+shiftedsymplectic}
\begin{flalign}
\omega_{-1}^\Sigma\big((C^\dagger,c^\dagger), C \big)\,&=\, 
\int_{\Sigma_+} C^\dagger\wedge C + \int_{(\partial \Sigma)_+}c^\dagger \wedge \iota^\ast C\quad,\\
\omega_{-1}^{\Sigma}\big((A^\dagger,a^\dagger,\varphi^\dagger),(A,\varphi)\big)\,&=\,
 \int_{\Sigma_+} A^\dagger\wedge A + \int_{(\partial \Sigma)_+} \big(a^\dagger\wedge \iota^\ast A + \varphi^\dagger\wedge\varphi\big)\quad,
\end{flalign}
\end{subequations}
for all $(C^\dagger,c^\dagger)\in \Omega_\cc^m(M)\times \Omega_{\cc}^{m-1}(\partial M)$,
$C \in \Omega^0_\cc(M)$, $(A^\dagger,a^\dagger,\varphi^\dagger)\in \Omega_\cc^{m-1}(M)
\times \Omega_{\cc}^{m-2}(\partial M) 
\times \Omega_{\cc}^{m-1}(\partial M)$ and $(A,\varphi)\in \Omega^1_\cc(M)\times\Omega^0_{\cc}(\partial M)$.
It is important to emphasize that, in contrast to the $[-1]$-shifted symplectic structure 
in \eqref{eqn:shiftedsymplectic},
the restricted integrations in \eqref{eqn:+shiftedsymplectic} {\em do not} define a chain map,
i.e.\ the pre-composition $\omega_{-1}^{\Sigma} \circ \dd^\otimes \neq 0$ with the differential
$\dd^\otimes$ of the tensor product chain complex $\Sol_\cc(M)\otimes\Sol_\cc(M)$ is non-zero.
However, we obtain a chain map $\omega_{-1}^{\Sigma} \circ \dd^\otimes  : 
\Sol_\cc(M)\otimes\Sol_\cc(M)\to\bbR$ to the unshifted real numbers, because
the differential $\dd^\otimes$ has degree $-1$ and the chain map property 
$\omega_{-1}^{\Sigma} \circ \dd^\otimes\circ \dd^\otimes =0 $ is a consequence 
of nilpotency ${\dd^\otimes}^2=0$ of the differential.
\sk

We are now in a position to define the {\em unshifted} 
symplectic structure associated with a Cauchy surface $\Sigma$.
\begin{defi}\label{def:unshiftedsymplectic}
The unshifted symplectic structure associated with $\Sigma\subset M$
is the chain map
\begin{flalign}
\omega_0^{\Sigma} \,:=\, \omega_{-1}^{\Sigma}\circ \dd^{\otimes}\,:\,\Sol_\cc(M)\otimes\Sol_\cc(M)~\longrightarrow~\bbR\quad.
\end{flalign}
\end{defi}
\begin{propo}\label{propo:unshiftedsymplectic}
The unshifted symplectic structure is explicitly given by
\begin{subequations}\label{eqn:unshiftedsymplectic}
\begin{flalign}
\label{eqn:YMsymplectic}
\omega^\Sigma_0\big((A,\varphi), (A^\prime,\varphi^\prime)\big)\,&=\,
\int_{\Sigma} \big(A\wedge \ast \dd A^\prime - A^\prime\wedge \ast \dd A\big)
- \int_{\partial\Sigma} \big(\varphi \wedge \ast_{\partial} \dd_{A^\prime} \varphi^\prime - 
\varphi^\prime \wedge \ast_{\partial} \dd_A \varphi\big)\quad,\\
\label{eqn:ghostsymplectic}
\omega^\Sigma_0\big((A^\dagger,a^\dagger ,\varphi^\dagger), C \big) \,&=\,
(-1)^m \int_{\Sigma}A^\dagger\wedge C - (-1)^{m-1}\int_{\partial\Sigma}a^\dagger \wedge \iota^\ast C\quad,
\end{flalign}
\end{subequations}
for all  $(A,\varphi), (A^\prime,\varphi^\prime) \in \Omega^1_\cc(M)\times\Omega^0_{\cc}(\partial M)$,
$(A^\dagger,a^\dagger ,\varphi^\dagger)\in \Omega_\cc^{m-1}(M)\times \Omega_{\cc}^{m-2}(\partial M)
\times \Omega_{\cc}^{m-1}(\partial M)$ and $C \in \Omega^0_\cc(M)$.
\end{propo}
\begin{proof}
The proof is a straightforward calculation using Stokes' theorem for manifolds with boundaries and corners, see
e.g.\ \cite{StokesCorner}. Thus, we will not write out the details of this calculation.
However, for the benefit of the reader, we note that there are two different instances
of Stokes' theorem that enter this calculation (consider the picture in \eqref{eqn:integrationpicture} 
for a helpful visualization). 
First, for any $\zeta\in\Omega^{m-1}_\cc(\Sigma_+)$ 
in the bulk $\Sigma_+$, Stokes' theorem with corners yields
\begin{flalign}
\int_{\Sigma_+}\dd \zeta \,=\, \int_{\Sigma}\zeta + \int_{(\partial \Sigma)_+} \iota^\ast \zeta\quad,
\end{flalign}
because $\partial(\Sigma_+) = (\partial \Sigma)_+ \cup \Sigma$. Second, for any
$\eta\in\Omega^{m-2}_\cc((\partial \Sigma)_+) $ on the time-like boundary component $(\partial \Sigma)_+$, 
ordinary Stokes' theorem yields
\begin{flalign}
\int_{(\partial \Sigma)_+}\dd \eta \,= \, - \int_{\partial \Sigma} \eta\quad,
\end{flalign}
because $\partial((\partial \Sigma)_+) = -\partial\Sigma$ is the boundary of $\Sigma$ with the opposite orientation.
\end{proof}
\begin{cor}\label{cor:omega0sc}
Using the same formulas as in \eqref{eqn:unshiftedsymplectic},
the unshifted symplectic structure from Definition \ref{def:unshiftedsymplectic}
and Proposition \ref{propo:unshiftedsymplectic} admits an extension
to a chain map
\begin{flalign}\label{eqn:omega0sc}
\omega_{0}^\Sigma\,:\, \Sol_{\sc}(M)\otimes\Sol_{\sc}(M)~\longrightarrow~\bbR\quad,
\end{flalign}
where $\Sol_{\sc}(M)$ is the space-like compactly supported
analog of the solution complex \eqref{eqn:solcomplex}. 
(Recall that a differential form $\zeta\in \Omega^p(M)$ has space-like compact support if 
$\supp(\zeta) \subseteq J^+_M(K) \cup J^-_M(K)$, for some compact subset $K\subseteq M$.)
\end{cor}

\begin{rem}\label{rem:extendedphasespace}
At first sight, it seems that our unshifted symplectic structure \eqref{eqn:unshiftedsymplectic}
is different from the one proposed in \cite{DF16}. However, upon closer inspection, one finds that this
is {\em not} the case and that the $0$-truncation of our approach reproduces the results of \cite{DF16}.
Let us recall that \cite{DF16} are not working in a homological approach, which means that they are 
implementing the Euler-Lagrange equations \eqref{eqn:ELequations} in the strict sense. From our
perspective, this means that they are considering $0$-cycles in the space-like
compactly supported solution complex $\Sol_{\sc}(M)$. For every two $0$-cycles
$(A,\varphi),(A^\prime,\varphi^\prime)\in \Omega^1_\sc(M)\times \Omega^0_{\sc}(\partial M)$,
i.e.\ $P(A,\varphi)=0=P(A^\prime,\varphi^\prime)$ with $P$ given in \eqref{eqn:Poperator},
one can write the unshifted symplectic structure \eqref{eqn:YMsymplectic}
equivalently as
\begin{flalign}
\nn \omega_0^\Sigma\big((A,\varphi), (A^\prime,\varphi^\prime)\big)\,&=\,\int_{\Sigma} \big(A\wedge \ast \dd A^\prime - A^\prime \wedge \ast \dd A\big)- \int_{\partial\Sigma} \big(\varphi\wedge \ast_{\partial} \dd_{A^\prime} \varphi^\prime 
- \varphi^\prime \wedge \ast_{\partial} \dd_A \varphi\big)\\
\,&=\, \int_{\Sigma} \big(A\wedge \ast \dd A^\prime - A^\prime \wedge \ast \dd A\big) 
-\int_{\partial\Sigma} \big(\varphi \wedge\iota^\ast(\ast \dd A^\prime) - \varphi^\prime 
\wedge \iota^\ast(\ast \dd A)\big)\quad,\label{eqn:DFsymplecticcorner}
\end{flalign}
where we used explicitly the matching constraint from the Euler-Lagrange equations \eqref{eqn:ELequations}.
This equivalent form of the unshifted symplectic structure on $0$-cycles coincides
with the proposal in \cite{DF16}. We note that the antifield-ghost component 
\eqref{eqn:ghostsymplectic} of our unshifted symplectic structure is a novel feature
of our homological approach that has no corresponding analog in the $0$-truncation 
studied in \cite{DF16}.
\end{rem}

\begin{rem}\label{rem:CMRlinear}
We would like to conclude this section with a comparison
of our results to the BV-BFV formalism \cite{Cattaneo,Mnev}.
Because the study of electromagnetism in \cite[Section 5.1]{Cattaneo}
does not include edge modes (in contrast to the newer study in \cite[Section 4]{Mnev}, 
on which we will comment below), we shall focus first on the case of an empty boundary 
$\partial M=\emptyset$.
The solution complex \eqref{eqn:solcomplex} then simplifies to
\begin{flalign}
\Sol(M) \,=\, \Big(
\xymatrix@C=4.5em{
\stackrel{(-2)}{\Omega^m(M)}
& \ar[l]_-{(-1)^m\,\dd} \stackrel{(-1)}{\Omega^{m-1}(M)}
& \ar[l]_-{(-1)^{m-1}\,\dd\ast\dd} \stackrel{(0)}{\Omega^1(M)}
& \ar[l]_-{\dd} \stackrel{(1)}{\Omega^0(M)}
}
\Big)
\end{flalign}
and the $[-1]$-shifted symplectic structure \eqref{eqn:shiftedsymplectic} simplifies to
\begin{flalign}
\omega_{-1}\big(C^\dagger, C \big)\,=\, 
\int_M C^\dagger\wedge C \quad,\qquad
\omega_{-1}\big(A^\dagger , A \big)\,=\,
 \int_M A^\dagger\wedge A \quad.
\end{flalign}
Furthermore, the unshifted symplectic structure \eqref{eqn:unshiftedsymplectic} associated
with a Cauchy surface $\Sigma\subset M$ simplifies to
\begin{flalign}
\omega^\Sigma_0\big(A, A^\prime\big)\,=\,
\int_{\Sigma} \big(A\wedge \ast \dd A^\prime - A^\prime\wedge \ast \dd A\big)\quad,\qquad
\omega^\Sigma_0\big(A^\dagger, C \big) \,=\,
(-1)^m \int_{\Sigma}A^\dagger\wedge C \quad.
\end{flalign}
We observe that both the $[-1]$-shifted and unshifted symplectic structure
agree with the ones obtained from the BV-BFV formalism 
applied to electromagnetism \cite[Section 5.1]{Cattaneo}.
Furthermore, we obtain as in \cite[Section 5.1.6]{Cattaneo} a $[+1]$-shifted symplectic structure 
in codimension $2$ by iterating our construction in Definition \ref{def:unshiftedsymplectic}.
Concretely, let us choose any codimension $1$ submanifold $S\subset \Sigma$ 
of the Cauchy surface (i.e.\ $S\subset M$ is codimension $2$) 
and cut $\Sigma$ along $S$. This defines two submanifolds $S_+,S_-\subset \Sigma$ with boundary
$\partial(S_\pm) = \pm S$ which determine $\Sigma$ by pasting $\Sigma=S_+ \sqcup_{S} S_- $.
Analogously to \eqref{eqn:+shiftedsymplectic}, we define $\omega^{S}_0$
by restricting the integrations from $\Sigma$ to $S_+\subset \Sigma$.
The $[+1]$-shifted symplectic structure can then be defined
analogously to Definition \ref{def:unshiftedsymplectic} as
$\omega^S_{1} := \omega^{S}_0\circ \dd^\otimes : \Sol_\cc(M)\otimes\Sol_\cc(M)\to \bbR[1]$.
By a straightforward calculation using Stokes' theorem, we obtain
\begin{flalign}
\omega^S_{1}\big(A,C\big) \,=\, -\int_S \ast\dd A \wedge C\,=\,-\omega^S_{1}\big(C,A\big) \quad.
\end{flalign}
Note that this matches the codimension $2$ $[+1]$-symplectic structure 
in \cite[Section 5.1]{Cattaneo}.
Finally, by a further iteration of our construction 
in Definition \ref{def:unshiftedsymplectic},
one easily shows that the $[+2]$-shifted symplectic structure in codimension $3$ is zero.
\sk

The results in \cite[Section 4]{Mnev} generalize the BV-BFV formalism
for Yang-Mills theory to the case of a boundary $\partial M\neq \emptyset$
including edge modes. Unfortunately, a direct and explicit comparison to our results 
in this section seems to be difficult, because the quantities of interest to us, 
in particular the boundary action in \eqref{eqn:action} and the unshifted symplectic structure
\eqref{eqn:unshiftedsymplectic}, have not been worked out 
in \cite{Mnev} for the Yang-Mills example. We refer the reader 
to Remark \ref{rem:MSWcomparison} below, where we compare 
the results we obtain by applying our techniques to linear Chern-Simons theory
with the results from \cite{Mnev}, which are in this case more detailed than for the Yang-Mills example.
\end{rem}

%%%%%%%%%%%%%%%%%%%%%%%%%%%%%%%%%%%%%%%%%%%%%%%%
%%%%%%%%%%%%%%%%%%%%%%%%%%%%%%%%%%%%%%%%%%%%%%%%

\section{\label{sec:CS}Linear Chern-Simons theory}
In this last section we shall apply our techniques to 
investigate edge modes in linear Chern-Simons theory.
This will allow us to compare in more depth our approach 
to the one proposed in \cite{Geiller}, which is based on Donnelly and 
Freidel's methods \cite{DF16}, and the one in \cite{Mnev}, which is a 
generalization of the BV-BFV formalism \cite{Cattaneo}.
\sk

Let us fix a $3$-dimensional manifold
$M=\bbR\times\Sigma$ with smooth boundary $\partial M = \bbR\times\partial \Sigma$.
We assume that both $M$ and the $2$-dimensional manifold $\Sigma$ are oriented, 
hence the factor $\bbR$ is oriented too.
As an input for our construction, we have to specify in analogy
to Section \ref{sec:model} the following data:
1.) a groupoid of bulk fields on $M$, 2.) a boundary condition on $\partial M$, and
3.) a gauge-invariant action functional on the groupoid of fields
satisfying the boundary condition (in the sense of homotopy pullbacks, 
cf.\ \eqref{eqn:boundarycondition} and also Remark \ref{rem:originedgemodes}).
For our linear Chern-Simons model, we take the groupoid of bulk fields 
$\mathbf{B}\bbR_{\mathrm{con}}(M)$ from \eqref{eqn:BGcon} and implement
the same topological boundary condition \eqref{eqn:boundarycondition}
as in the case of linear Yang-Mills theory. Hence, the field groupoid $\FFF(M)$
is precisely the one of Proposition \ref{propo:fieldgroupoid}, see in particular \eqref{eqn:Fgroupoid}.
Instead of \eqref{eqn:action}, we propose now the following action
\begin{flalign}\label{eqn:CSaction}
S(A,\varphi) \,:= \, \int_M \frac{1}{2}\,A\wedge\dd A + \int_{\partial M}\frac{1}{2}\Big(  \dd \varphi \wedge \iota^\ast A + 
\lambda \, \dd_A\varphi \wedge \ast_{\partial}^{} \dd_A \varphi \Big)\quad,
\end{flalign}
where $\lambda\in\bbR$ is a parameter on which we shall comment later.
Note that, for defining the third term, we have chosen (the conformal class of) a Lorentzian metric $g_\partial^{}$ 
on the boundary $\partial M$. This term is necessary to reproduce the
well-known chiral currents on $\partial M$, see e.g\ \cite{Bieri}.
The second term in the action \eqref{eqn:CSaction} is needed to compensate the failure
of the usual Chern-Simons action $\int_M \frac{1}{2} \,A\wedge \dd A$ to be gauge-invariant 
in the presence of a boundary $\partial M\neq \emptyset$. For any choice of $\lambda\in\bbR$, 
the total action \eqref{eqn:CSaction} is invariant under the gauge transformations 
in $\FFF(M)$, see \eqref{eqn:Fgroupoid}.
Varying this action with respect to compactly supported variations 
$(\alpha,\psi)\in \Omega^1_\cc(M)\times \Omega^0_\cc(\partial M)$ yields
\begin{flalign}
\delta_{(\alpha,\psi)} S(A,\varphi) = \int_M\alpha\wedge \dd A - 
\int_{\partial M}\frac{1}{2}\Big( \iota^\ast \alpha \wedge\big(\dd_A \varphi + 2\lambda \ast_{\partial}^{}\dd_A\varphi\big)
+ \psi\wedge \big(2\lambda \, \dd \ast_{\partial}\dd_A\varphi + \iota^\ast (\dd A)\big)\Big)\quad.
\end{flalign}
The corresponding Euler-Lagrange equations are
\begin{subequations}
\begin{flalign}
\label{eqn:CSeqn}\dd A &\,=\,0 \qquad\text{(linear Chern-Simons equation on $M$)}\quad,\\
2 \lambda \, \dd\ast_{\partial}^{}\dd_A \varphi  + \iota^\ast(\dd A)&\,=\, 0 \qquad \text{(inhomogeneous Klein-Gordon equation 
on $\partial M$)}\quad,\\
\label{eqn:chirality}\dd_A\varphi +2\lambda\, \ast_\partial^{} \dd_{A}\varphi&\,=\,0\qquad
 \text{(matching constraint on $\partial M$)}\quad.
\end{flalign}
\end{subequations}
We observe that the first and third equations imply the second one, 
hence the independent equations of motion for our model are
given by \eqref{eqn:CSeqn} and \eqref{eqn:chirality}. Note that the latter equation
specializes to the (anti-)self-duality constraint $\ast_\partial \dd_A\varphi = \pm \dd_A \varphi$
for the $1$-form $\dd_A \varphi\in\Omega^1(\partial M)$, provided that we choose 
$\lambda = \mp \frac{1}{2}$. For our studies, we shall keep the parameter $\lambda$ arbitrary.
\sk

Because we are dealing with a linear field theory, we can reformulate the Chern-Simons model from above
in the language of chain complexes and compute as in Section \ref{sec:shifted}
the derived critical locus together with its $[-1]$-shifted symplectic structure.
Since the calculations are completely analogous to the case of linear Yang-Mills theory,
we shall present only the final results. The Chern-Simons solution complex is given by
\begin{subequations}
\begin{flalign}
\Sol(M) \,=\, \Big(
\xymatrix@C=2.5em{
\stackrel{(-2)}{\FFF_{1,\cc}(M)^\ast}
& \ar[l]_-{Q^\ast} \stackrel{(-1)}{\FFF_{0,\cc}(M)^\ast}
& \ar[l]_-{P } \stackrel{(0)}{\FFF_{0}(M)}
& \ar[l]_-{Q} \stackrel{(1)}{\FFF_{1}(M)}
}
\Big)\quad,
\end{flalign}
where
\begin{flalign}
\FFF_1(M) \,&=\,\Omega^0(M)\quad,\\
\FFF_0(M) \,&=\, \Omega^1(M)\times \Omega^0(\partial M)\quad,\\
\FFF_{0,\cc}(M)^\ast \,&=\, \Omega^2(M)\times \Omega^1(\partial M) \times \Omega^2(\partial M)\quad,\\
\FFF_{1,\cc}(M)^\ast \,&=\, \Omega^3(M)\times \Omega^2(\partial M)\quad,
\end{flalign}
and
\begin{flalign}
Q(C)\,&=\, \big(\dd C,\iota^\ast C\big)\quad,\\
P(A,\varphi) \,&=\, \Big(\dd A ,\frac{1}{2}\big(\dd_A\varphi + 2\lambda\,\ast_{\partial}^{}\dd_A\varphi \big) ,-\frac{1}{2}
\big(2\lambda\, \dd\ast_{\partial}^{}\dd_A \varphi + \iota^\ast (\dd A)\big) \Big)\quad,\\
Q^\ast(A^\dagger,a^\dagger,\varphi^\dagger)\,&=\, \Big(-\dd A^\dagger , \dd a^\dagger + \iota^\ast A^\dagger + 
\varphi^\dagger\Big)\quad.
\end{flalign}
\end{subequations}
The $[-1]$-shifted symplectic structure
$\omega_{-1} : \Sol_\cc(M)\otimes \Sol_\cc(M)\to \bbR[-1]$ for linear Chern-Simons
theory coincides with the one for linear Yang-Mills theory from \eqref{eqn:shiftedsymplectic}.
\sk

Let us now choose the  datum of a `Cauchy' surface $ \Sigma\subset M$,
by which we mean in this case a surface of constant $t_0\in \bbR$. Using 
the orientation on the $\bbR$-factor of $M=\bbR\times \Sigma$, we introduce 
the `future' $\Sigma_+\subseteq M$  of this surface and 
its intersection $(\partial \Sigma)_+ := \Sigma_+\cap \partial M$ with the boundary $\partial M$.
The geometric picture is again as in \eqref{eqn:integrationpicture}, where now the role of time
is played by the $\bbR$-factor of the product manifold $M=\bbR\times \Sigma$.
In analogy to Definition \ref{def:unshiftedsymplectic}, we define
the unshifted symplectic structure associated with the `Cauchy' surface $\Sigma\subset M$
by $\omega^\Sigma_0 := \omega^{\Sigma}_{-1}\circ \dd^\otimes :  \Sol_\cc(M)\otimes \Sol_\cc(M)\to \bbR$,
where we recall that $\omega^{\Sigma}_{-1}$ is the $[-1]$-shifted symplectic structure
with integrations restricted to $\Sigma_+$ and $(\partial\Sigma)_+$. By a 
direct calculation using Stokes' theorem (see the proof of Proposition \ref{propo:unshiftedsymplectic} for
some hints and instructions), we obtain
\begin{subequations}\label{eqn:CSsymplectic}
\begin{flalign}
\omega^\Sigma_0\big((A,\varphi),(A^\prime,\varphi^\prime)\big)\,&=\, 
\int_\Sigma A\wedge A^\prime - \int_{\partial\Sigma}\frac{1}{2} \Big(\varphi\wedge (2\lambda\,\ast_{\partial}^{}\dd_{A^\prime}\varphi^\prime +\iota^\ast A^\prime) -\varphi^\prime\wedge (2\lambda\,\ast_{\partial}^{}\dd_{A}\varphi +\iota^\ast A) \Big) \quad,\\
\omega^\Sigma_0\big((A^\dagger,a^\dagger,\varphi^\dagger),C\big) \,&=\, -\int_{\Sigma} A^\dagger \wedge C - \int_{\partial\Sigma} a^\dagger\wedge \iota^\ast C\quad.
\end{flalign}
\end{subequations}
When evaluated on two $0$-cycles, i.e.\ $P(A,\varphi)=0=P(A^\prime,\varphi^\prime)$,
one can write the unshifted symplectic structure equivalently as
\begin{flalign}\label{eqn:CSsymplectic0cycles}
\omega^\Sigma_0\big((A,\varphi),(A^\prime,\varphi^\prime)\big)
\,=\, \int_\Sigma A\wedge A^\prime + \int_{\partial\Sigma} \frac{1}{2}\big(\varphi\wedge \dd_{A^\prime}\varphi^\prime- \varphi^\prime \wedge \dd_A\varphi \big) -\int_{\partial\Sigma}\frac{1}{2}\big(\varphi\wedge\iota^\ast A^\prime - \varphi^\prime \wedge \iota^\ast A \big)
\end{flalign}
by using explicitly the matching  constraint \eqref{eqn:chirality}.
Let us recall and emphasize that the unshifted symplectic structure is by construction
invariant under gauge transformations $(A,\varphi)\to (A+\dd \epsilon,\varphi+\iota^\ast \epsilon)$
and $(A^\prime,\varphi^\prime)\to (A^\prime+\dd \epsilon,\varphi^\prime+\iota^\ast \epsilon)$ of
$0$-cycles. The first term in \eqref{eqn:CSsymplectic0cycles}
is the usual symplectic structure for linear Chern-Simons theory
and the second term is that of a chiral free boson on $\partial M$. The third term
is the analog for linear Chern-Simons theory of the corner contribution
to the linear Yang-Mills symplectic structure by Donnelly and Freidel \cite{DF16}, see also
\eqref{eqn:DFsymplecticcorner}. Note that our unshifted presymplectic structure 
\eqref{eqn:CSsymplectic0cycles} agrees with the proposal in \cite[Eqn.~(3.64)]{Geiller}.
(The apparent sign differences are due to Geiller's opposite sign
convention $(A,\varphi)\to (A+\dd \epsilon,\varphi-\iota^\ast \epsilon)$ for gauge 
transformations of the edge modes.)
\begin{rem}
From the expression in \eqref{eqn:CSsymplectic0cycles}, it seems that our unshifted symplectic
structure on $0$-cycles is independent of the choice of the free parameter $\lambda$
in the action \eqref{eqn:CSaction}. This is indeed true, provided that
we do not set $\lambda = 0$. In this special case, the matching constraint \eqref{eqn:chirality}
degenerates to $\dd_A\varphi =0$, hence the chiral free boson is eliminated from the field content 
and consequently its contribution to the symplectic structure \eqref{eqn:CSsymplectic0cycles} vanishes.
\end{rem}

\begin{rem}\label{rem:MSWcomparison}
We would like to conclude by comparing
our results to the ones obtained within the BV-BFV formalism 
\cite{Cattaneo,Mnev}, see in particular \cite[Section 2.7]{Mnev} for 
the example of interest to us. We first observe that the boundary 
action in \cite[Eqn.\ (85)]{Mnev}, which is obtained by a transgression construction 
and the choice of polarization functional in \cite[Eqn.\ (83)]{Mnev}, is related to our choice of
boundary action in \eqref{eqn:CSaction}: Using that the Hodge operator 
$\ast_{\partial}^{2} =\id$ squares to the identity on $\Omega^1(\partial M)$, 
we can decompose $\iota^\ast A = A_+ + A_-$ 
and $\dd\varphi = \dd_+\varphi + \dd_-\varphi$ into self-dual and anti-self-dual parts.
Inserting this decomposition into \eqref{eqn:CSaction}, one obtains 
\begin{multline}
S(A,\varphi) \,=\,\int_M\frac{1}{2}\,A\wedge\dd A  \\
+ \int_{\partial M}\frac{1}{2}\Big((2\lambda+1)\,  \dd_+\varphi\wedge A_- 
+ (2\lambda-1) \, A_+ \wedge \dd_-\varphi 
+ 2\lambda \, A_-\wedge A_+  + 2\lambda \, \dd_- \varphi\wedge \dd_+\varphi\Big)\quad,
\end{multline}
which in the self-dual case $\lambda = -\frac{1}{2}$ and in the anti-self-dual case
$\lambda=\frac{1}{2}$ reduces to boundary actions analogous to \cite[Eqn.\ (85)]{Mnev}.
\sk

According to our best understanding, the paper \cite{Mnev} does not seem
to study unshifted symplectic structures for the edge modes; at least no
results were explicitly stated. Hence, a comparison to our unshifted symplectic 
structure \eqref{eqn:CSsymplectic}, whose $0$-truncation \eqref{eqn:CSsymplectic0cycles} agrees with the 
proposal by Geiller \cite{Geiller}, is not possible at the moment.
We however believe that, for practitioners of the BV-BFV formalism, 
it should be possible to obtain analogous results
in the framework proposed in \cite{Mnev}.
\end{rem}

%%%%%%%%%%%%%%%%%%%%%%%%%%%%%%%%%%%%%%%%%%%%%%%%
%%%%%%%%%%%%%%%%%%%%%%%%%%%%%%%%%%%%%%%%%%%%%%%%

\section*{Acknowledgments}
We would like to thank Marco Benini, William Donnelly, Laurent Freidel,
Jorma Louko, Pavel Mnev and Konstantin Wernli for useful discussions and 
comments on this work. We also would like to thank the anonymous referee 
for their valuable comments and suggestions that helped us to improve this manuscript.
L.M.\ gratefully acknowledges the support of NSF grant DMS-1547292.
A.S.\ gratefully acknowledges the financial support of 
the Royal Society (UK) through a Royal Society University 
Research Fellowship, a Research Grant and an Enhancement Award.
N.T.\ gratefully acknowledges the support of NSF grant STS-1734155.

\section*{Conflict of interest statement}
On behalf of all authors, the corresponding author states that there is no conflict of interest. 

%%%%%%%%%%%%%%%%%%%%%%%%%%%%%%%%%%%%%%%%%%%%%%%%
%%%%%%%%%%%%%%%%%%%%%%%%%%%%%%%%%%%%%%%%%%%%%%%%

\appendix

\section{\label{app:hPB}Homotopy pullback constructions}
The aim of this appendix is to provide more details on
the homotopy pullback constructions for groupoids
and chain complexes that are used in Propositions
\ref{propo:fieldgroupoid} and  \ref{propo:solcomplex}.
Generally speaking, homotopy pullbacks are higher categorical
generalizations of pullbacks from ordinary category theory
that are compatible not only with isomorphisms, but also with the appropriate concepts of
weak equivalences in these contexts, e.g.\ categorical equivalences
in the category of groupoids $\Grpd$ or
quasi-isomorphisms in the category of chain complexes $\Ch_\bbR$.
A general theory of homotopy limits (and colimits) can be developed
in the framework of model category theory \cite{Hovey} by making use 
of derived functors. For the purpose of our work, however, 
we do not have to focus too much on these abstract considerations as it will be
sufficient to provide and explain explicit models for computing homotopy 
pullbacks for groupoids and chain complexes.
\sk

Let us start with the case of groupoids. Let $f : \mathcal{G} \to \mathcal{K}$
and $g : \mathcal{H}\to\mathcal{K}$ be two functors between groupoids
$\mathcal{G},\mathcal{H},\mathcal{K}\in\Grpd$ and consider the homotopy 
pullback diagram
\begin{flalign}\label{eqn:homotopyPBgrpd}
\xymatrix{
\ar@{-->}[d]\mathcal{P} \ar@{-->}[r] &  \ar@{}[dl]_-{h~~} \mathcal{H}\ar[d]^-{g}\\
\mathcal{G} \ar[r]_-{f}& \mathcal{K}
}
\end{flalign}
The following explicit description of the groupoid $\mathcal{P}\in\Grpd$ that is 
determined by this homotopy pullback is well known, see e.g.\ \cite[Section~2]{Hollander}.
\begin{propo}\label{prop:homotopyPBgrpd}
A model for the homotopy pullback in \eqref{eqn:homotopyPBgrpd} is given by the
groupoid $\mathcal{P}$ whose
\begin{itemize}
\item[(i)] objects are triples $(x,y,k)$ with $x\in \mathcal{G}$, $y\in \mathcal{H}$
and $k : f(x)\to g(y)$ an isomorphism in $\mathcal{K}$, and
\item[(ii)] morphisms are pairs $(\phi,\psi) : (x,y,k)\to (x^\prime,y^\prime,k^\prime)$
with $\phi : x\to x^\prime$ a morphism in $\mathcal{G}$ and $\psi: y\to y^\prime$
a morphism in $\mathcal{H}$, such that the diagram
\begin{flalign}
\xymatrix@C=4em{
\ar[d]_-{k} f(x) \ar[r]^-{f(\phi)} & f(x^\prime)\ar[d]^-{k^\prime}\\
g(y) \ar[r]_-{g(\psi)} & g(y^\prime)
}
\end{flalign}
in $\mathcal{K}$ commutes.
\end{itemize}
\end{propo}
\begin{rem}
For illustrative purposes, let us consider the special case where all groupoids 
$\mathcal{G},\mathcal{H},\mathcal{K}\in \Set\subseteq \Grpd$ are sets,
i.e.\ every morphism in these groupoids is an identity morphism.
Then the groupoid $\mathcal{P}$ from Proposition \ref{prop:homotopyPBgrpd}
is a set too, namely
\begin{flalign}
\mathcal{P} \,\cong\, \big\{(x,y)\in \mathcal{G}\times\mathcal{H} \,:\, f(x)=g(y)\big\}\in\Set\quad.
\end{flalign}
Observe that this is the ordinary pullback (also called fiber product) in the category of
sets, which consists of pairs of elements $(x,y)\in \mathcal{G}\times\mathcal{H}$
whose images in $\mathcal{K}$ under $f$ and $g$ coincide.
The general homotopy pullback for groupoids from Proposition 
\ref{prop:homotopyPBgrpd} admits a similar interpretation:
Its objects are pairs of objects $(x,y)\in \mathcal{G}\times\mathcal{H}$
together with {\em an isomorphism $k : f(x)\to g(y)$ witnessing that
$f(x)$ and $g(y)$ `coincide' in $\mathcal{K}$ in the sense that they are isomorphic}.
A morphism in $\mathcal{P}$ can then be interpreted as a pair of morphisms $\phi : x\to x^\prime$
and $\psi : y\to y^\prime$ that is compatible with these witnesses.
\end{rem}

Let us consider now the case of chain complexes $\Ch_\bbR$. To simplify our presentation, 
we shall focus only on the specific class of homotopy pullbacks
that is needed for computing linear derived critical loci as in Sections \ref{sec:shifted}
and \ref{sec:CS}. We refer to \cite[Section~3]{Walter} for a study of more general types of homotopy 
pullbacks and also other homotopy (co)limits.
Let $V,W\in \Ch_\bbR$ be two chain complexes.
We assume that $V$ is non-negatively graded, i.e.\ $V_n=0$ 
for all $n<0$, and that $W$ non-positively graded,
i.e.\ $W_n=0$ for all $n>0$. (This assumption is always 
satisfied in applications to derived critical loci for linear gauge field theories, where
$V$ is a chain complex encoding the gauge fields and (higher) ghost fields
in non-negative degrees, and $W = V^\ast$ is the dual of $V$. We would like to emphasize 
that, in this context, the antifields are not included in $V$, but they are a result of the 
homotopy pullback construction in Proposition \ref{prop:homotopyPBchain} below.) 
We regard the projection chain map
$\pi_1 : V\times W \to V$ on the first factor of the Cartesian product
complex as a bundle over $V$ with fiber $W$ and consider two sections,
the zero-section $(\id,0) : V\to V\times W$ and a generic section $(\id,f) : V\to V\times W$,
where $f :V\to W$ is any chain map. Because $V$ is by hypothesis
non-negatively graded and $W$ is non-positively graded, the chain map
$f$ is necessarily of the form
\begin{flalign}
\xymatrix{
\ar[d]_-{0}\cdots ~&~ \ar[d]_-{0}\ar[l]_-{0}0 ~&~ \ar[d]_-{f_0} V_0 \ar[l]_-{0} ~&~ \ar[d]_-{0} V_1 \ar[l]_-{\dd^V}~&~ \ar[l]_-{\dd^V} \cdots\ar[d]_-{0}\\
\cdots ~&~ W_{-1} \ar[l]^-{\dd^W}~&~ W_0 \ar[l]^-{\dd^W}~&~ 0 \ar[l]^-{0}~&~ \ar[l]^-{0}\cdots
}
\end{flalign}
i.e.\ it is determined by a single linear map $f_0 : V_0\to W_0$ in degree $0$ 
that has to satisfy $f_0 \circ \dd^V =0$ and $\dd^W \circ f_0=0$.
Our goal is to provide an explicit description of the chain complex $L\in \Ch_\bbR$ that is 
determined by the homotopy pullback
\begin{flalign}\label{eqn:homotopyPBchain}
\xymatrix@C=3.5em{
\ar@{-->}[d]L \ar@{-->}[r] &  \ar@{}[dl]_-{h~~~~~~~} V\ar[d]^-{(\id,f)}\\
V \ar[r]_-{(\id,0)}& V\times W
}
\end{flalign}
\begin{propo}\label{prop:homotopyPBchain}
A model for the homotopy pullback in \eqref{eqn:homotopyPBchain} is given by the chain complex
\begin{flalign}\label{eqn:LcomplexTMP}
L \,=\, \Big( \xymatrix{
\cdots ~&~ \stackrel{(-2)}{W_{-1}} \ar[l]_-{-\dd^W}~&~ \stackrel{(-1)}{W_0} \ar[l]_-{-\dd^W}~&~ \stackrel{(0)}{V_0} \ar[l]_-{f_0}~&~ \stackrel{(1)}{V_1} \ar[l]_-{\dd^V}~&~ \ar[l]_-{\dd^V} \cdots
}\Big)\quad,
\end{flalign}
where we indicate in round brackets the homological degrees of $L$.
\end{propo}
\begin{proof}
We follow the same strategy as in \cite[Proposition~3.21]{BSreview},
where a special case of this proposition was proven. In particular, we will
compute the homotopy pullback \eqref{eqn:homotopyPBchain} in terms
of an {\em ordinary} pullback by replacing the zero-section $(\id,0) : V\to V\times W$
by a weakly equivalent fibration. To construct such a replacement, let us introduce
the chain complex
\begin{flalign}
D \,:=\, \Big(
\xymatrix{
\stackrel{(-1)}{\bbR} ~&~ \ar[l]_-{\id}\stackrel{(0)}{\bbR}
}
\Big)
\end{flalign}
concentrated in degrees $0$ and $-1$. Note that this complex is acyclic,
i.e.\ the unique map $0 \to D$ from the zero complex to $D$ is a quasi-isomorphism.
Let us consider now the tensor product complex $D\otimes W$, which is 
acyclic too, and observe that it is explicitly given by
$(D\otimes W)_{n} \cong W_n \oplus W_{n+1} \cong W_n\times W_{n+1}$,
for all $n\in\bbZ$, together with the differential
\begin{flalign}\label{tmp:differential}
\dd^{D\otimes W}\big(w_n, \widetilde{w}_{n+1}\big) \,=\, \big(\dd^W w_n , w_n - \dd^W \widetilde{w}_{n+1}\big) \quad,
\end{flalign}
for all $(w_n,\widetilde{w}_{n+1})\in W_n\times W_{n+1}$. We define
a chain map $p: D\otimes W \to W$, sending $(w_n,\widetilde{w}_{n+1}) \mapsto w_n$,
and note that this map is degree-wise surjective and hence a fibration in $\Ch_\bbR$.
With these preparations, we obtain a fibration
\begin{flalign}
\id\times p \,:\, V\times(D\otimes W) ~\longrightarrow~ V\times W
\end{flalign}
that is a weakly equivalent replacement of the zero-section $(\id,0) : V\to V\times W$.
This allows us to compute the chain complex $L$ in \eqref{eqn:homotopyPBchain} 
by the ordinary pullback
\begin{flalign}
\xymatrix@C=3.5em{
\ar@{-->}[d]L \ar@{-->}[r] &  V\ar[d]^-{(\id,f)}\\
V\times (D\otimes W) \ar[r]_-{\id\times p}& V\times W
}
\end{flalign}
Explicitly, the degree $n\in\bbZ$ component of the chain complex $L$ is given by
\begin{subequations}\label{tmp:Lcomplex}
\begin{flalign}
\nn L_n \,:=&\, \Big\{ \big( (v_n ,w_n ,\widetilde{w}_{n+1}) , v_n^\prime\big) \in V_n\times W_n\times W_{n+1}\times V_n\, : \, (v_n,w_n) = (v_n^\prime,f(v_n^\prime))\Big\}\\
\cong& ~ V_n\times W_{n+1}\quad,
\end{flalign}
where the isomorphism in the last step is given by $(v_n,\widetilde{w}_{n+1}) \mapsto
\big((v_n, f(v_n) , \widetilde{w}_{n+1}), v_n\big)$. Using \eqref{tmp:differential}, we can compute
the  differential $\dd^{L} : L_n\to L_{n-1}$ and find
\begin{flalign}
\dd^L\big(v_n , \widetilde{w}_{n+1}\big)\,=\, \big(\dd^V v_n , f(v_n) - \dd^W \widetilde{w}_{n+1}\big)\quad,
\end{flalign}
\end{subequations}
for all $(v_n,\widetilde{w}_{n+1}) \in V_n\times W_{n+1}$.
Recalling that $V$ is by hypothesis non-negatively graded and $W$ is non-positively graded,
we find that the chain complex described by \eqref{tmp:Lcomplex} is isomorphic to 
the chain complex in \eqref{eqn:LcomplexTMP}, which completes our proof.
\end{proof}

%%%%%%%%%%%%%%%%%%%%%%%%

\end{document}